\documentclass[10pt,conference]{IEEEtran}

\usepackage{graphicx}
\usepackage{url}
\usepackage{epstopdf}
\usepackage{amsfonts,amssymb,amsbsy, amsthm}
\usepackage{amsmath}
\usepackage{times}
\usepackage{bbm}
\usepackage{tikz}
\usetikzlibrary{arrows, positioning}

\newtheorem{theorem}{Theorem}

\newtheorem{lemma}{Lemma}
\newtheorem{proposition}{Proposition}
\newtheorem{definition}{Definition}
\newtheorem{remark}{Remark}

\tikzset{
    >=stealth',
    user/.style={
           rectangle,
           rounded corners,
           draw,
           minimum height=2em,
           minimum width=1cm,
           text centered},
    relay/.style={
           rectangle,
           rounded corners,
           draw,
           minimum height=1em,
           minimum width=1cm,
           text centered},
    pil/.style={
           ->,
           shorten <=2pt,
           shorten >=2pt},
    pil_rev/.style={
           <-,
           shorten <=2pt,
           shorten >=2pt},
    pil_dash/.style={
    		->, dashed,
    		shorten <=2pt,
    		shorten >=2pt}
}

\DeclareMathOperator{\SNR}{\text{SNR}}
\DeclareMathOperator*{\argmin}{arg\,min}

\begin{document}

\title{A Low-Complexity Message Recovery Method\\for Compute-and-Forward Relaying}

\author{Amaro Barreal$^\dagger$, Joonas P\"a\"akk\"onen$^\dagger$, David Karpuk$^\dagger$, Camilla Hollanti$^\dagger$, Olav Tirkkonen$^\ddag$
\\
$^\dagger$Department of Mathematics and Systems Analysis, School of Science, Aalto University. \\
$^\ddag$Department of Communications and Networking, School of Electrical Engineering, Aalto University. \\
\{amaro.barreal, joonas.paakkonen, david.karpuk, camilla.hollanti, olav.tirkkonen\}@aalto.fi}

\maketitle

\begin{abstract}
The Compute-and-Forward relaying strategy achieves high computation rates by decoding linear combinations of transmitted messages at intermediate relays. However, if the involved relays independently choose which combinations of the messages to decode, there is no guarantee that the overall system of linear equations is solvable at the destination. In this article it is shown that, for a Gaussian fading channel model with two transmitters and two relays, always choosing the combination that maximizes the computation rate often leads to a case where the original messages cannot be recovered. It is further shown that by limiting the relays to select from carefully designed sets of equations, a solvable system can be guaranteed while maintaining high computation rates. The proposed method has a constant computational complexity and requires no information exchange between the relays.
\end{abstract}

\section{Introduction}
In wireless multiuser relay networks, both interference from multiple transmitters and noise degrade the system performance. To combat these issues, Nazer and Gastpar recently introduced a new relaying strategy called Compute-and-Forward (CaF) \cite{nazer}. Their key idea is to decode an integer linear combination of the transmitted messages at intermediate relays, and then forward the combinations to the destination.

Finding integer combinations that yield high transmission rates turns out, however, to be a complicated task. Particularly, finding the equation coefficients of the linear combinations that maximize the data transmission rate coincides with a Shortest Vector Problem (SVP) \cite{lll}, for which various algorithms have been proposed \cite{fengargmin,soussi,wei}. Unfortunately, these algorithms tend to be either highly complex or suboptimal.
Recently, algorithms of polynomial complexity for finding the equation coefficients that maximize the rate have been presented \cite{svpsahraei,svpwen}. However, choosing the coefficient vectors that maximize the instantaneous computation rate might result in an overall unsolvable system of linear equations at the destination due to linear dependency of the coefficient vectors -- that is, the original messages might not necessarily be recoverable even if the combinations are successfully decoded at the relays.

The message recoverability problem has been addressed in \cite{soussi}, where precoding at the transmitters is used to increase the probability of receiving independent combinations at the destination, while \cite{wei,pappi,mejri} allow cooperation between the relays.
These methods either require preprocessing at the transmitters or signaling between the relays.

In \cite{hong}, this problem is mitigated by choosing a subset of relays with suitable equation coefficients. This approach can however not be used in the symmetric case, that is if the number of relays equals the number of transmitters.

In this paper, we introduce a new, efficient approach to finding the coefficient vectors for a system with two transmitters and two relays. We observe that, in a Gaussian fading channel, there are only few coefficient vectors that typically maximize the computation rate. Based on this observation, we compile small coefficient vector candidate sets for both relays so that the probability of the relays choosing linearly dependent vectors vanishes.

Our proposed method does not require jointly finding the coefficients at the relays, thus there is no need for cooperation. Furthermore, searching for appropriate coefficient vectors only over a small set of vectors reduces the computational complexity as opposed to solving the corresponding SVP. 

The main benefits of our proposed method are twofold. Firstly, the end-to-end information outage probability vanishes as the message rate at the transmitters approaches zero. Secondly, the computational complexity of our scheme is constant while still providing a relatively high throughput. Our findings are supported by extensive computer simulations.

The paper is organized as follows. We give a brief introduction of the CaF protocol in Section \ref{compu}, and introduce methods for finding suitable equation coefficients in Section \ref{vectorsec}. Section \ref{simu} presents the performance metrics and corresponding numerical results, while Section \ref{conclu} concludes the paper.

\section{The Compute-and-Forward Protocol}\label{compu}
In this article, we focus on a wireless multiple-access system with $L=2$ transmitters and $M=2$ relays, as illustrated in Figure \ref{systemfig}. 
\begin{figure}[h]
\centering
\begin{tikzpicture}
	\node[relay] (r1) {\scriptsize Relay 1};
	\node[below=0.07 of r1] (vert2) {};
	\node[relay,below=0.05 of vert2] (rn) {\scriptsize Relay 2};
	\node[left=3 of vert2] (vert1) {};
	\node[user, above=0.05 of vert1] (u1) {\scriptsize Tx 1}
		edge[pil_dash] node[pos=0.3,above]{\tiny $h_{11}$} (r1.west)
		edge[pil_dash] node[pos=0.3,above]{\tiny $h_{21}$} (rn.west);
	\node[user, below=0.07 of vert1] (uk) {\scriptsize Tx 2}
		edge[pil_dash] node[pos=0.3,below]{\tiny $h_{12}$} (r1.west)
		edge[pil_dash] node[pos=0.3,below]{\tiny $h_{22}$} (rn.west);
	\node[user, right=2.5 of vert2] (dest) {\scriptsize Dest.}
		edge[pil_rev] (r1.east)
		edge[pil_rev] (rn.east);
	\node[right=1 of u1] (dummy1) {};
	\node[above=0.2 of dummy1] (1hop) {\scriptsize First Hop};
	\node[right=1.8 of 1hop] (2hop) {\scriptsize Second Hop};
\end{tikzpicture}
\caption{System model with two transmitters and two relays connected to a destination. The first hop is modeled as a wireless multiple access fading channel. The relays are connected to the destination with error-free bit pipes.}
\label{systemfig}
\end{figure}
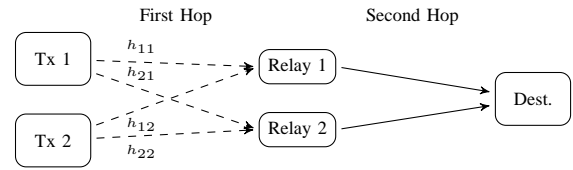
The first hop from the transmitters to the relays is modeled as a Gaussian fading channel. The relays are connected to a destination with error-free bit pipes with unlimited capacities. The goal of the system is to reliably transfer information from both transmitters, via the relays, all the way to the destination. The relays apply the original CaF strategy introduced in \cite{nazer}, briefly exposed in the following. 

The transmitters want to communicate messages $\textbf{w}_l \in \mathbb{F}_p^k$, $l=1,2$, where $p$ is prime. Before transmission, these messages are encoded into $n$-dimensional codewords $\mathbf{x}_l$ that are subject to the power constraint $||\mathbf{x}_l||^2\leq nP$.
Note that throughout this paper, we assume that both transmitters use identical transmission powers, and that the relays have identical noise levels with normalized variance, \emph{i.e.}, $\SNR = P$. This is justified as allowing asymmetric $\SNR$ values affects the performance metric used for simulations for all three considered strategies equally, and hence does not affect the comparison results.

\begin{definition}
The \emph{message rate} at transmitter $l$ is defined as
\begin{align}
\mathcal{R}_l^s = \frac{k}{n}\log_2 p.
\end{align}
\end{definition}

For the received signal at relay $m$, we use the following channel model:
\begin{align}
\mathbf{y}_m = \sum_{l=1}^L h_{ml}\mathbf{x}_l + \mathbf{z}_m,  
\end{align}
where $\mathbf{z}_m$ is additive white Gaussian noise with normalized variance $\sigma^2=1$, and the channel coefficients are assumed to be i.i.d. and real-valued, $h_{ml} \sim \mathcal{N}(0,1)$ \cite{nazer}. Let $\mathbf{h}_m=\left[\begin{smallmatrix} h_{m1} & h_{m2} & \cdots & h_{mL}\end{smallmatrix}\right]^T$ denote the channel vector for relay $m$. 

Channel state information is only available at the relays. Further, it is important to note that we assume that each relay only knows the channels to itself, i.e., relay 1 knows $h_{11}$ and $h_{12}$, while relay 2 knows $h_{21}$ and $h_{22}$.

The key feature of CaF is to decode and forward linear combinations of the transmitted messages. In this paper, we only consider integer combinations of the transmitted codewords. 

\begin{remark}
	As noted above, each transmitter encodes its message $\textbf{w}_l \in \mathbb{F}_p^k$ into a lattice vector $\mathbf{x}_l$. Hence, although the destination ultimately wants to decode the messages $\mathbf{w}_l$, the relays decode linear equations involving the codewords $\mathbf{x}_l$, and it is thus meaningful to consider $\mathbb{Z}$-linear rather than $\mathbb{F}_p$-linear combinations. The destination, upon reception of enough $\mathbb{Z}$-linearly independent equations, can solve for the codewords and recover the original messages \cite{nazer}.  
\end{remark}

The integer combination of relay $m$ is represented by an \emph{equation coefficient vector} $\mathbf{a}_m$, more explicitely
\begin{align}
\mathbf{a}_m = \left[\begin{smallmatrix} a_{m1} & a_{m2} & \cdots & a_{mL}\end{smallmatrix}\right]^T.
\label{adef}
\end{align}
These vectors form the \emph{equation coefficient matrix}
\begin{align}
\mathbf{A}=\left[\begin{smallmatrix} \mathbf{a}_{1} & {\bf a}_{2} & \cdots & {\bf a}_{M}\end{smallmatrix}\right]^T,
\label{matrix}
\end{align}
and the destination can recover the single codewords if and only if $\mathbf{A}$ is invertible, that is $\det(\mathbf{A}) \neq 0$. 

Figure~\ref{deterrorsfig} shows that, if the relays always choose the equation coefficient vector that maximizes their instantaneous computation rate, the probability of the matrix $\mathbf{A}$ being singular is noticeable. This phenomenon is more pronounced for low $\SNR$ values, and even more considerable for more than two relays and transmitters. In this article, we focus on the case $M = L = 2$.  

\begin{figure}[h]
		\includegraphics[scale=0.32]{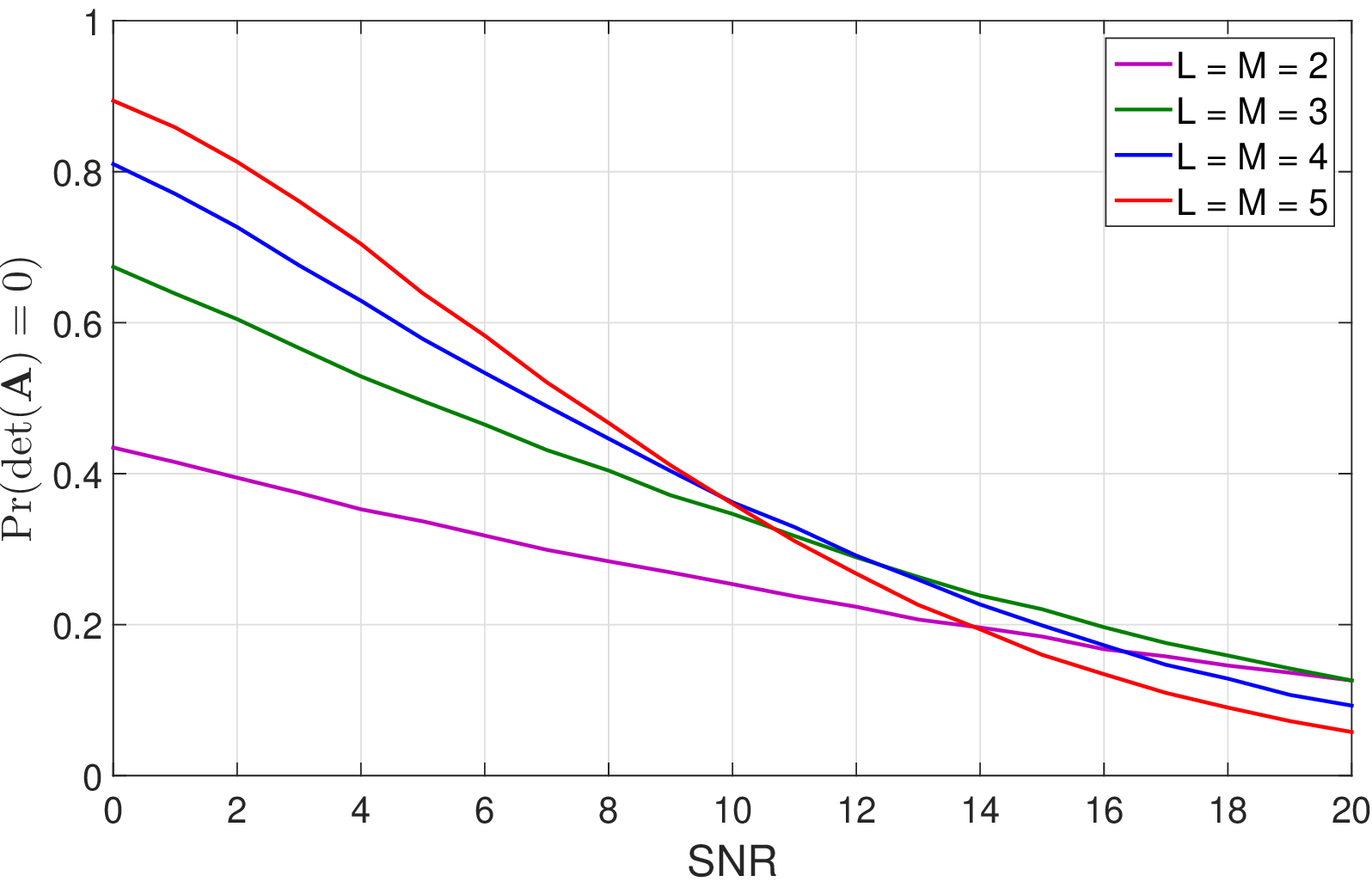}
		\caption{Probability of the equation coefficient matrix $\mathbf{A}$ being singular for various numbers of relays $M$ and transmitters $L$ when independently choosing the coefficient vectors that maximize the computation rate.}
		\label{deterrorsfig}
\end{figure}

The most important performance metric of CaF is the so-called \emph{computation rate}, whose meaning is explained in the following. If relay $m$ achieves computation rate $\mathcal{R}_m^r = \mathcal{R}_m^r(\mathbf{h}_m,\mathbf{a}_m)$, then it is able to decode a linear combination of the codewords whose corresponding message rates satisfy $\mathcal{R}^s \le \mathcal{R}_m^r$. That is, relay $m$ can support message rate $\mathcal{R}_m^r$ if
\begin{align}
{\cal R}_{l}^{s}\leq\min_{a_{ml}\neq 0}{\cal R}_{m}^{r},
\end{align}
where ${\cal R}_{l}^{s}$ is the message rate at transmitter $l$, and $a_{ml}\neq 0$ means that the message of transmitter $l$ is included in the linear combination decoded by relay $m$.

In the scenario considered in this paper, we have two transmitters with identical message rates ${\cal R}^{s}={\cal R}_{1}^{s}={\cal R}_{2}^{s}$.

Both relays can support this message rate if
\begin{align}
{\cal R}^{s} \le \min\{\mathcal{R}_1^r,\mathcal{R}_2^r\}.
\end{align}

The main observations about the computation rate derived in \cite{nazer} are summarized in the following theorem.
\begin{theorem} 
	In the above setup, a computation rate region of
\begin{align}
	\mathcal{R}_m^r(\mathbf{h}_m,\mathbf{a}_m) = \max\limits_{\alpha \in \mathbb{R}} \frac{1}{2}\log^+\left(\frac{P}{\alpha^2+P||\alpha\mathbf{h}_m-\mathbf{a}_m||^2}\right)
\end{align}
is achievable. This expression is further maximized by choosing $\alpha$ to be 
\begin{align}
	\alpha_{\text{MMSE}} = \frac{P\mathbf{h}_m^T\mathbf{a}_m}{1+P||\mathbf{h}_m||^2},
\end{align}
resulting in a computation rate region of 
\small
\begin{align}
\mathcal{R}_m^r(\mathbf{h}_m,\mathbf{a}_m) = \frac{1}{2}\log^+\left(\left(||\mathbf{a}_m||^2-\frac{P|\mathbf{h}_m^T\mathbf{a}_m|^2}{1+P||\mathbf{h}_m||^2}\right)^{-1}\right).
\label{compurate}
\end{align}
\end{theorem}
The computation rate is thus a function of the channel coefficients and the equation coefficients. While the channel coefficients are random variables, the relays are free to choose their desired equation coefficients. In this paper, we mainly focus on how to efficiently find equation coefficient vectors which yield high computation rates while ensuring that the vectors chosen by the relays are linearly independent.

\section{Selecting the Equation Coefficient Vectors}\label{vectorsec}
In this section we present three methods in detail for finding equation coefficient vectors $\mathbf{a}_m$ at the relays, all of which require no cooperation between the relays.

The main interest in this article is to guarantee a solvable system of equations at the destination, while still being able to support high message rates. To that end, we will include the following indicator function in our performance metric: 
\begin{align}
\label{indicator}
	\mathbbm{1}_{\left\{\det(\mathbf{A}) \neq 0\right\}} = \begin{cases} 1 &\mbox{if } \det(\mathbf{A}) \neq 0, \\ 0 &\mbox{if } \det(\mathbf{A}) = 0. \end{cases} 
\end{align}

\emph{1. Exhaustive search:} Each relay searches over all possible equation coefficient vectors and chooses the one that maximizes its instantaneous computation rate \eqref{compurate}.

\begin{lemma} \cite{fengargmin}
Finding the coefficient vector that maximizes the instantaneous computation rate is equivalent to solving
\begin{align}
\mathbf{a}_m = \argmin_{\mathbf{a}\in\mathbb{Z}^L\backslash\left\{(0,\ldots,0)\right\}}\mathbf{a}^T \mathbf{Ga},
\label{argmineq}
\end{align}
where 
\begin{align}
	\mathbf{G} = \mathbf{I}_{L}-\frac{P\mathbf{h}_m\mathbf{h}_m^T}{1+P||\mathbf{h}_m||^2},
\end{align}
and corresponds to finding the shortest vector in the lattice whose Gram matrix is $\mathbf{G}$.  
\end{lemma}

 While maximizing the computation rate individually at each relay yields the highest message rate supported by the given channel realizations, there is no guarantee that the chosen vectors form an invertible matrix at the destination. 

\emph{2. Proposed splitting method:} In our proposed method, we split possible coefficient vectors into two disjoint sets $V_1$ and $V_2$. Set $V_1$ is pre-assigned to one relay, while set $V_2$ is pre-assigned to the other. Each relay independently chooses the vector within its set that maximizes the computation rate \eqref{compurate}. 

\begin{remark}
\label{quadratic_form}
	Let $\mathbf{a}_m = \left[\begin{smallmatrix} a_{m1} & a_{m2}\end{smallmatrix}\right]^T$ be a solution to \eqref{argmineq}. Note that the matrix $\mathbf{G}$ is symmetric, that is \eqref{argmineq} results in 
	\[
		\mathbf{a}_m^T \mathbf{G} \mathbf{a}_m = \mathbf{a}_m^T \left[\begin{smallmatrix} x & z \\ z & y \end{smallmatrix}\right] \mathbf{a}_m = a_{m1}^2x + 2a_{m1} a_{m2} z + a_{m2}^2 y.
	\]
Thus if $\left[\begin{smallmatrix} a_{m1} & a_{m2} \end{smallmatrix}\right]^T$ is a solution to \eqref{compurate}, so is $\left[\begin{smallmatrix} -a_{m1} & -a_{m2} \end{smallmatrix}\right]^T$. 
In particular, we can fix $a_{m1} \ge 0$. 

One may thus divide the set of candidate vectors $\mathbb{Z}^2\backslash\left\{(0,0)\right\}$ into equivalence classes modulo the $\mathbb{Z}_2$-action of $\pm 1$.  Solutions to \eqref{argmineq} are sought for in the  quotient space $\mathcal{E} = (\mathbb{Z}^2\backslash\left\{(0,0)\right\}) \mod \mathbb{Z}_2$. The elements in $\mathcal{E}$ are equivalence classes $\left[\mathbf{a}\right]$ of vectors $\mathbf{a}\in \mathbb{Z}^2 \backslash\left\{(0,0)\right\}$ up to $\mathbb{Z}_2$, i.e. pairs of vectors, $\left[\mathbf{a}\right] = \left\{\mathbf{a},-\mathbf{a}\right\}$. A unique way of representing elements in $\mathcal{E}$ is by points in $\mathbb{Z}^2\backslash\left\{(0,0)\right\}$ in the right half plane, including the upper y-axis, but not the lower part, \emph{i.e.}, with vectors $(a_1, a_2)$, $a_1>0$ or $a_1=0$, $a_2>0$.
\end{remark}

\begin{proposition}\label{prop:rotation}
The rotation $\mathbf{U} = \left[\begin{smallmatrix} 0 & 1 \\ -1 & 0 \end{smallmatrix}\right]$ acts in a non-degenerate fashion on $\mathcal{E}$. Any point $\left[\mathbf{a}\right] \in \mathcal{E}$ has a unique pair $\left[\mathbf{b}\right]\neq \left[\mathbf{a}\right] \in \mathcal{E}$, so that $\mathbf{U}\left[\mathbf{a}\right] = \left[\mathbf{b}\right]$ and $\mathbf{U}\left[\mathbf{b}\right] = \left[\mathbf{a}\right]$. Thus $\mathcal{E}$ can be divided into disjoint sets $\mathcal{E}_1$ and $\mathcal{E}_2$ so that $\mathcal{E}_1 \cup \mathcal{E}_2 = \mathcal{E}$, $\mathcal{E}_1\cap \mathcal{E}_2 = \emptyset$, and $\mathbf{U}\mathcal{E}_1 = \mathcal{E}_2$.
\end{proposition}

\begin{proof}
	Consider the vector $\mathbf{a} = (a_1, a_2) \in \mathbb{Z}^2\backslash\left\{(0,0)\right\}$. Its equivalence class is $\left\{(a_1,a_2),(-a_1,-a_2)\right\}$. The action of $\mathbf{U}$ on a vector commutes with scalar multiplication by $\pm 1$. Thus $\mathbf{U}\left[\mathbf{a}\right] = \left\{(a_2, –a_1),(-a_2, a_1)\right\}  \equiv \left[\mathbf{b}\right] \in \mathcal{E}$. Clearly, $\left[\mathbf{a}\right]\neq \left[\mathbf{b}\right]$, as $\mathbf{a}\neq (0,0)$, and $\mathbf{U}\left[\mathbf{b}\right] = \left\{(-a_1,-a_2),(a_1,a_2)\right\} = \left[\mathbf{a}\right]$. The action of $\mathbf{U}$ thus divides $\mathcal{E}$ into a disjoint set of pairs. One element of each pair can be taken to $\mathcal{E}_1$, the other to $\mathcal{E}_2$.
\end{proof}

Following Remark \ref{quadratic_form}, we define the following set.  
\begin{align}
\label{seteqn}
	\begin{split}
	V = \left\{\left.(x,y) \in \mathbb{Z}^2 \right| x \ge 1, y \neq 0, x^2 + y^2 \le d^2, \right. \\
	\left. d \in \mathbb{Z}\backslash\left\{0\right\}, \gcd(x,y) = 1\right\}.
	\end{split}
\end{align}
Let $\tilde{V}_1, \tilde{V}_2 \subset V$ such that $|\tilde{V}_1| = |\tilde{V}_2|$, $\tilde{V}_1 \cup \tilde{V}_2 = V$ and $\tilde{V}_1 \cap \tilde{V}_2 = \emptyset$, and define
\begin{align}
\label{setseqn2}
	V_1 = \tilde{V}_1 \cup \left\{(0,1)\right\}, \quad
	V_2 = \tilde{V}_2 \cup \left\{(1,0)\right\}.
\end{align} 

While preventing the relays from selecting arbitrary coefficient vectors decreases the expected computation rate at the relays, choosing linearly independent vectors ensures that the messages can be recovered at the destination -- provided that the relays are able to support the message rate at the transmitters. This claim is justified in the following proposition. 

\begin{proposition}
\label{prop:nonzero_rate}
	For any choice of sets $V_1$, $V_2$ as proposed above, we have almost surely 
	\begin{align}
		\min{\{\mathcal{R}_1^r,\mathcal{R}_2^r\}}\cdot \mathbbm{1}_{\left\{\det(\mathbf{A}) \neq 0\right\}} > 0.
	\end{align} 
\end{proposition}
\begin{proof}
	Since $V_1 \cap V_2 = \emptyset$, we have $\mathbbm{1}_{\left\{\det(\mathbf{A}) \neq 0\right\}} \equiv 1$. 
	It hence suffices to show that there exist vectors $\mathbf{a}_1 \in V_1$, $\mathbf{a}_2 \in V_2$ such that for every channel realization $\mathbf{h} = \left[\begin{smallmatrix} h_{1} & h_{2} \end{smallmatrix}\right]^T$,
	\begin{align*}
		\left(||\mathbf{a}_m||^2-\frac{P|\mathbf{h}^T \mathbf{a}_m|^2}{1+P||\mathbf{h}||^2}\right)^{-1} > 1,
	\end{align*}
	or equivalently $0 < ||\mathbf{a}_m||^2-\frac{P|\mathbf{h}^T \mathbf{a}_m|^2}{1+P||\mathbf{h}||^2} < 1$. 
	
	Choose $\mathbf{a}_1 = \left[\begin{smallmatrix} 1 & 0 \end{smallmatrix}\right]^T \in V_1$, $\mathbf{a}_2 = \left[\begin{smallmatrix} 0 & 1 \end{smallmatrix}\right]^T \in V_2$. Then, for $m = 1,2$ we have $||\mathbf{a}_m	||^2 = 1$, and, since 
	\begin{align*}
		|\mathbf{h}^T \mathbf{a}_m|^2 = \begin{cases} h_1^2 &\mbox{if } m = 2, \\ h_2^2 &\mbox{if } m = 1, \end{cases}
	\end{align*}
	for $m = 1,2$ we have $0 < \frac{P h_m^2}{1+P(h_1^2+h_2^2)} < 1$, hence 
	\begin{align*}
		0 < 1-\frac{P h_m^2}{1+P(h_1^2+h_2^2)} \le 1,
	\end{align*}
	as required, where equality holds if and only if $h_m = 0$. 		
\end{proof}

The following proposition gives a further criterion for designing the sets of vectors $V_1$ and $V_2$. 

\begin{proposition}
\label{prop:equal_avgrate}
	Let $V_1$ and $V_2$ be two sets of vectors, chosen as above, and denote by $\mathcal{R}_{V_1}^r$, $\mathcal{R}_{V_2}^r$ the expected computation rate of the first and second relay using the assigned sets, respectively. If there exists a rotation matrix $\mathbf{U}$ such that
	\begin{align}
		\mathbf{U}: V_1 \to V_2; \quad \mathbf{a} \mapsto \mathbf{U}\mathbf{a}
	\end{align}
	is a bijective isometry, we have
	\begin{align}
		\mathcal{R}_{V_1}^r = \mathcal{R}_{V_2}^r. 
	\end{align} 
\end{proposition}
\begin{proof}	
Assume $\mathbf{U}$ is such a matrix. Then, for every $\mathbf{a} \in V_2$ we find $\tilde{\mathbf{a}} \in V_1$ such that $\mathbf{a} = \mathbf{U}\tilde{\mathbf{a}}$. Moreover, $||\mathbf{a}||^2 = ||\mathbf{U}\tilde{\mathbf{a}}||^2$. 

Let $\mathbf{a} = \mathbf{U}\tilde{\mathbf{a}}$, $\tilde{\mathbf{h}} = \mathbf{U}^{T} \mathbf{h}$. Then, for a fixed power $P$, and since $\mathbf{U}^T \mathbf{U} = \mathbf{I}_2$,
	\begin{align*}
		\mathcal{R}^r(\mathbf{h},\mathbf{a}) &= \frac{1}{2}\log^{+}\left(\left(||\mathbf{a}||^2 - \frac{P|\mathbf{h}^T \mathbf{a}|^2}{1+P||\mathbf{h}||^2}\right)^{-1}\right) \\
		&= \frac{1}{2}\log^{+}\left(\left(||\mathbf{U}\tilde{\mathbf{a}}||^2 - \frac{P|\mathbf{h}^T \mathbf{U}\tilde{\mathbf{a}}|^2}{1+P||\mathbf{h}||^2}\right)^{-1}\right) \\
		&= \frac{1}{2}\log^{+}\left(\left(||\mathbf{U}\tilde{\mathbf{a}}||^2 - \frac{P|(\mathbf{U}\tilde{\mathbf{h}})^T \mathbf{U}\tilde{\mathbf{a}}|^2}{1+P||\mathbf{U}\tilde{\mathbf{h}}||^2}\right)^{-1}\right) \\
		&= \frac{1}{2}\log^{+}\left(\left(||\tilde{\mathbf{a}}||^2 - \frac{P|\tilde{\mathbf{h}}^T \tilde{\mathbf{a}}|^2}{1+P||\tilde{\mathbf{h}}||^2}\right)^{-1}\right) \\
		&= \mathcal{R}^r(\tilde{\mathbf{h}},\tilde{\mathbf{a}}).
	\end{align*}
	It follows $\mathcal{R}^r(\mathbf{h},\mathbf{a}) = \mathcal{R}^r(\mathbf{U}\tilde{\mathbf{h}},\mathbf{U}\tilde{\mathbf{a}}) = \mathcal{R}^r(\tilde{\mathbf{h}},\tilde{\mathbf{a}})$. 
	
	To conclude the proof, note that the distribution of the channel $\mathbf{h}$ is rotation invariant. We thus have
	\begin{align*}
		\mathcal{R}^r_{V_1} &= E_{\mathbf{h}}\left[\max\limits_{\mathbf{a} \in V_1}{\mathcal{R}^r(\mathbf{h},\mathbf{a})}\right] = E_{\mathbf{h}}\left[\max\limits_{\mathbf{U}\tilde{\mathbf{a}} \in V_1}{\mathcal{R}^r(\mathbf{h},\mathbf{U}\tilde{\mathbf{a}})}\right] \\
		&= E_{\mathbf{h}}\left[\max\limits_{\tilde{\mathbf{a}} \in V_2}{\mathcal{R}^r(\mathbf{U}\tilde{\mathbf{h}},\mathbf{U}\tilde{\mathbf{a}})}\right] = E_{\tilde{\mathbf{h}}}\left[\max\limits_{\tilde{\mathbf{a}}\in V_2}{\mathcal{R}^r(\tilde{\mathbf{h}},\tilde{\mathbf{a}})}\right] \\
	&= \mathcal{R}^r_{V_2}.
	\end{align*}
\end{proof}

\emph{3. Multiple Access:} For the sake of comparison, we also consider a simple multiple access method where each relay decodes the message corresponding to the stronger channel coefficient while treating the other message as noise. More concretely, if $\mathbf{h}_m = \left[\begin{smallmatrix} h_{m1} & h_{m2} \end{smallmatrix}\right]^T$ is the channel observed by relay $m$, then the relay decodes $\mathbf{x}_1$ iff $h_{m1} \ge h_{m2}$, and decodes $\mathbf{x}_2$ otherwise. 
In this scenario, for $i,j = 1,2$, $i \neq j$, if $|h_{mi}| \ge |h_{mj}|$, message $\mathbf{x}_i$, can be decoded by relay $m$ if
\begin{align}
	\mathcal{R}_i^s < \frac{1}{2}\log\left(1 + \frac{P|h_{mi}|^2}{1+P|h_{mj}|^2}\right).
\end{align}

We say that a determinant error occurs if both relays decode and forward the same message.

\section{Simulation Results}\label{simu}
In this section, we present extensive simulation results for the end-to-end performance of the methods introduced in Section~\ref{vectorsec} for the 2-by-2 relaying system described earlier. 

We start by defining two sets of vectors $V_1$ and $V_2$ as in \eqref{setseqn2} for the proposed strategy. Note that for a given channel realization $\mathbf{h}_m$, the equation coefficient vector $\mathbf{a}_m$ maximizing the computation rate at relay $m$ is the one that best aligns with $\mathbf{h}_m$ \cite{nazer}. It is thus crucial to provide each relay with vectors that can approximate channel vectors lying in any direction. 

Further, to achieve low computation complexity, we construct disjoint sets of small cardinality. For the following simulations, we fix $d = 4$ in \eqref{seteqn}, resulting in sets of cardinality $|V| = 14$, $|V_1| = |V_2| = 8$. 

\begin{proposition}\label{prop:splitting}
Splitting as discussed above exists, where $V_2 = \mathbf{U}V_1$. For such splittings,  $\min\left\{\mathcal{R}^r_1,\mathcal{R}^r_2\right\}\cdot\mathbbm{1}_{\left\{\det(\mathbf{A})\right\}} > 0$ for any channel realization, almost surely. Further, $\mathcal{R}^r_{V_1} = \mathcal{R}^r_{V_2}$. 
\end{proposition}
\begin{proof}
	Choose $V_1$ such that $\forall \mathbf{a}\in V_1$, either $\left[\mathbf{a}\right] \in \mathcal{E}_1$ or $\mathbf{U}\left[\mathbf{a}\right] \in \mathcal{E}_1$, but not both. Then according to Prop.~\ref{prop:rotation}, $\mathbf{U}V_1 \in \mathcal{E}_2$, and, $V_1 \cap V_2 = \emptyset$. Note that if $\left[(1,0)\right]\in \mathcal{E}_1$, we have $\mathbf{U}\left[(1,0)\right]  = \left[(0,1)\right] \in \mathcal{E}_2$. Thus, according to Prop.~\ref{prop:nonzero_rate}, the first statement holds and by Prop.~\ref{prop:equal_avgrate}, the second statement follows. 
\end{proof}

One example of such $V_1$ and $V_2$ is depicted in Figure~\ref{vectorfig}. 
\begin{figure}[h]
		\includegraphics[scale=0.26]{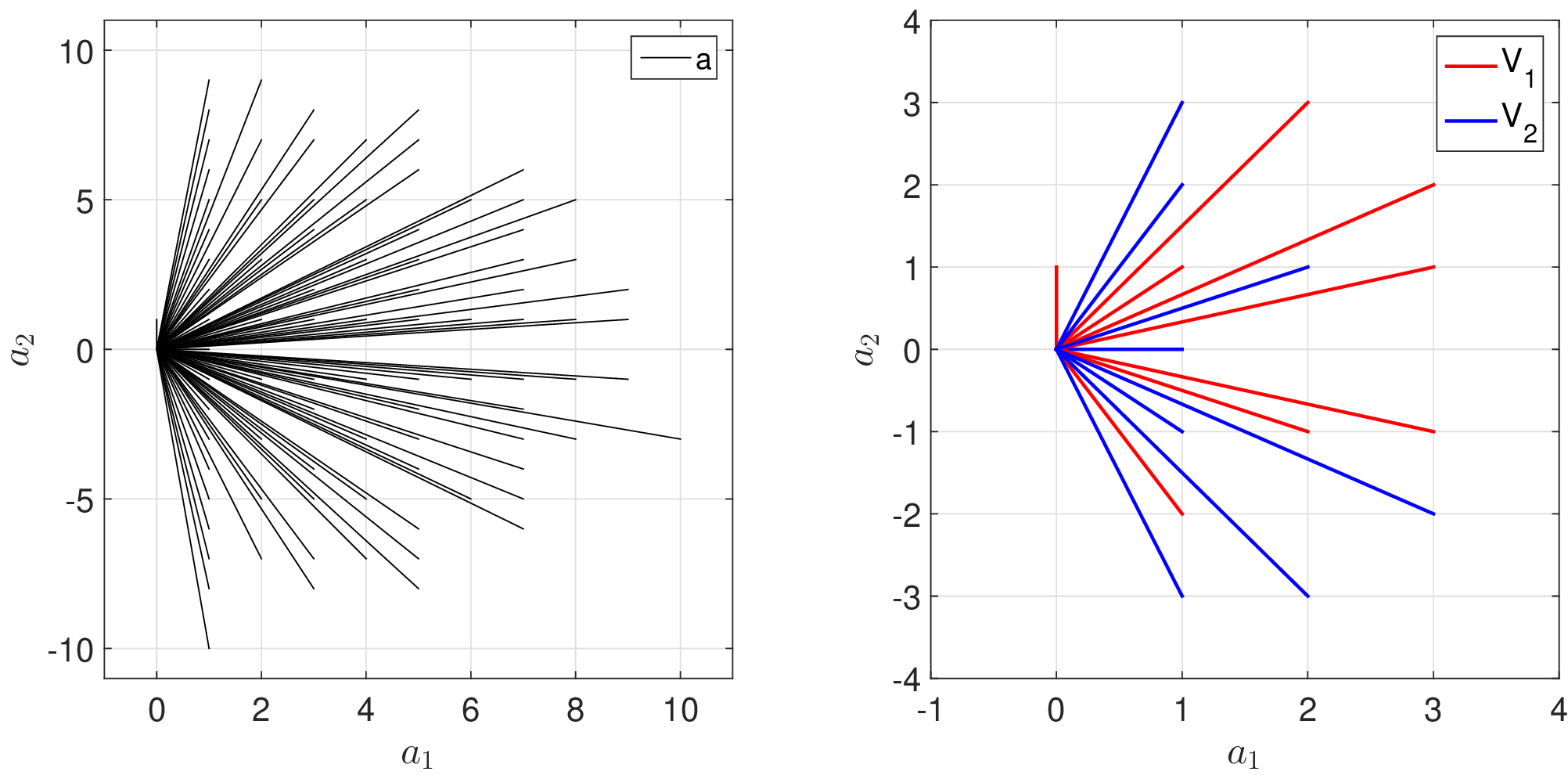}
		\caption{All vectors that maximize the computation rate for $10^5$ channel realizations (left) and sets of vectors $V_1$, $V_2$ for the proposed method (right).}
		\label{vectorfig}
\end{figure}

We now compare the three introduced methods. Our two main performance metrics are \emph{minimum end-to-end rate} and \emph{end-to-end outage} as described in the following.

\emph{Minimum end-to-end rate:} This metric is defined to be the expected value of the minimum of the two computation rates at the relays. If the coefficient matrix is not invertible, this metric is defined to be zero as the destination cannot recover both of the original messages. The minimum end-to-end rate can be expressed as
\begin{align}\label{minrate}
\mathcal{R}_0=E\left(\min{\{\mathcal{R}_1^r,\mathcal{R}_2^r\}}\cdot\mathbbm{1}_{\left\{\det(\mathbf{A}) \neq 0\right\}}\right),
\end{align}
where $E\left(\cdot\right)$ denotes the expected value, and $\mathcal{R}_m^r$ is the computation rate at relay $m$ defined as in \eqref{compurate}. We illustrate this performance metric for varying $\SNR$ values in Figure~\ref{expmin}. 
\begin{figure}[h]
		\includegraphics[scale=0.30]{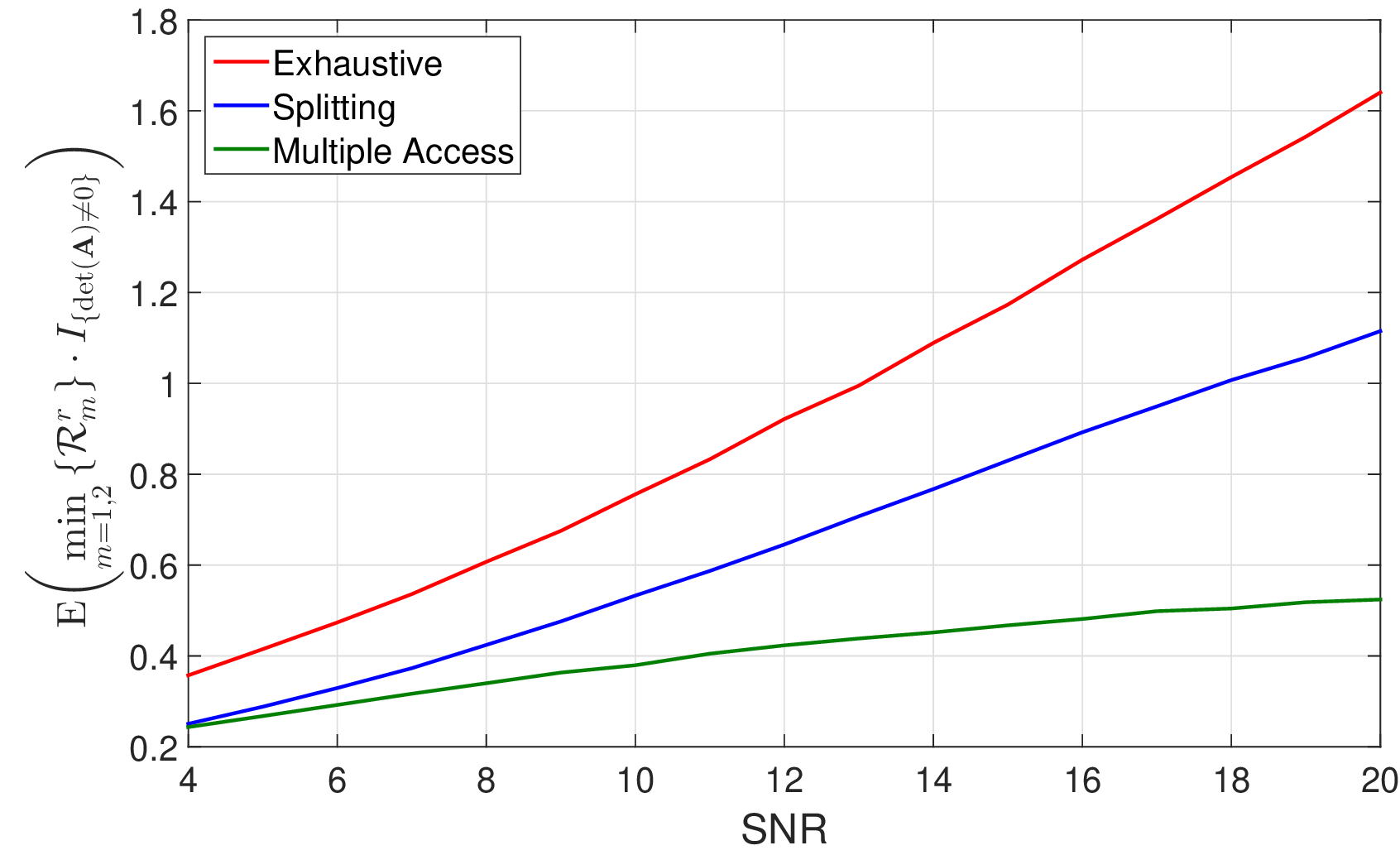}
		\caption{Expected minimum end-to-end rate vs. SNR.}
		\label{expmin}
\end{figure}

The proposed splitting method performs relatively close to the exhaustive search method, and clearly beats the multiple access strategy (cf. Section~\ref{vectorsec}-3). It is to be expected that the splitting method cannot outperform the exhaustive search whenever it has positive rate (cf. \eqref{minrate}) in terms of end-to-end rate, since the exhaustive search has a considerably higher number of coefficient vectors from which to choose. The drawback of the exhaustive search is, however, its high complexity compared to the constant complexity of our proposed method. We argue that, in many cases, it is more important to have a fast, lightweight algorithm that finds relatively high rates, rather than strictly maximizing the rate at the cost of computation time and power.

\emph{End-to-end outage:} The system is said to be in outage if at least one of the relays cannot support the symmetric message rate at the transmitters, or if the equation coefficient matrix is not invertible. The outage probability with symmetric message rates $\mathcal{R}^s$ can be expressed as
\begin{align}
\rho=\Pr{\left((\min{\{\mathcal{R}_1^r,\mathcal{R}_2^r\}} < \mathcal{R}_s) \cup (\det(\mathbf{A}) = 0)\right)}.
\end{align}

Figure \ref{outagefig} presents simulated end-to-end outage probabilities for a fixed $\SNR$ value as a function of the symmetric message rates. 
\begin{figure}[h]
		\includegraphics[scale=0.30]{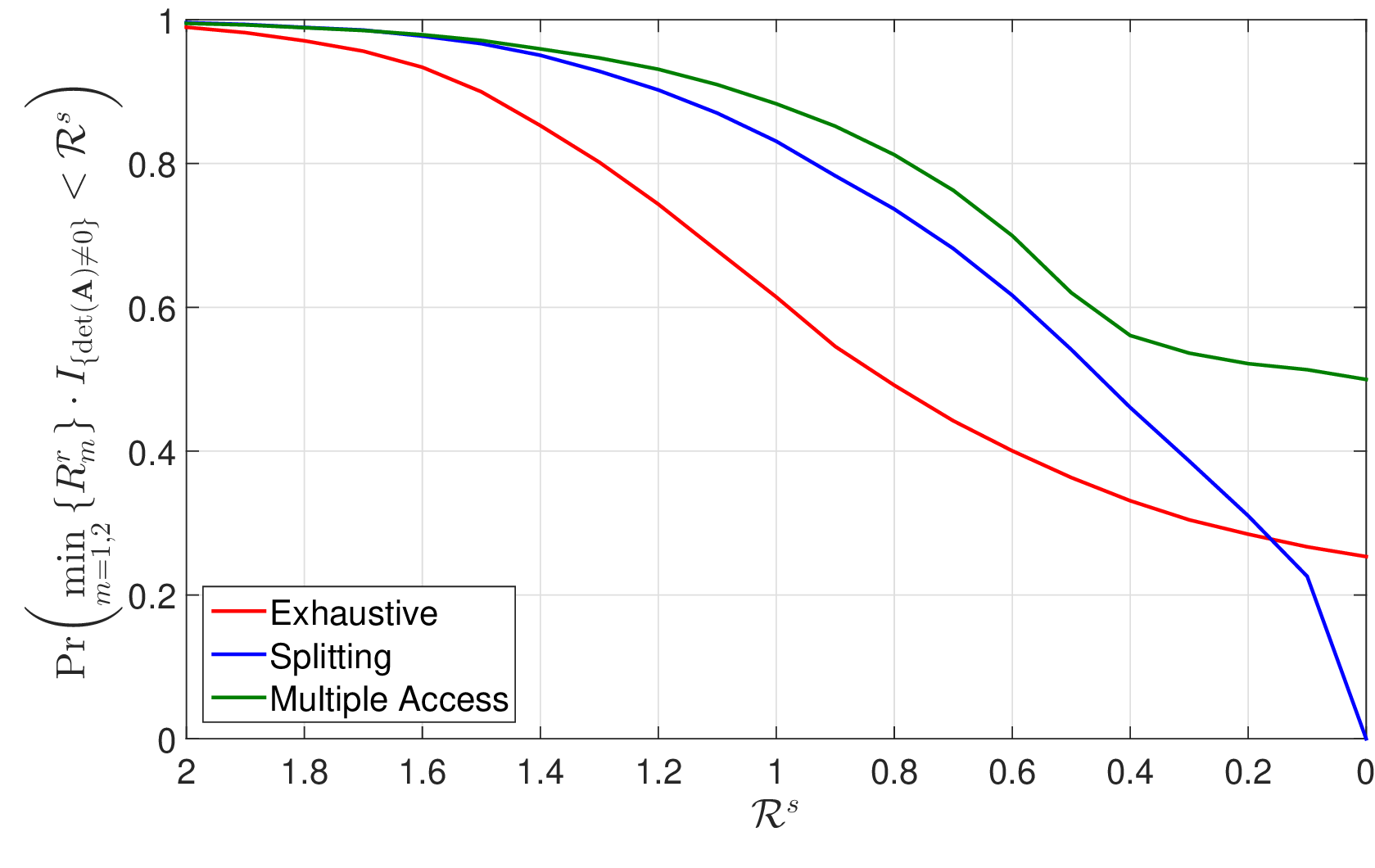}
		\caption{End-to-end outage as a function of the message rate for SNR 10 dB.}
		\label{outagefig}
\end{figure}
While for the exhaustive and multiple-access strategies the outage probability converges to a non-zero value for decreasing message rate, for the proposed splitting method it goes to zero as the message rate goes to zero, as indicated by Prop.~\ref{prop:nonzero_rate}.

\section{Conclusions and Future Work}\label{conclu}
In this article it was shown that pre-assigning sets of equation coefficient vectors to intermediate relays in the framework of Compute-and-Forward relaying helps preventing the overall system of linear equations at the destination from being singular, while maintaining high computation rates and thus being able to support high message rates. The advantage of the proposed method its constant complexity, as opposed to solving a hard shortest vector problem. Moreover, no cooperation between the relays is required.

In this article, we restricted ourselves to real-valued channels, as well as the case of only two transmitters and relays. It was however observed that the problem of the coefficient matrix being singular is even more dramatic for larger number of transmitters and relays. Therefore, as a natural extension, future work includes generalizing the introduced method for an arbitrary, not necessarily symmetric number of transmitters and relays, as well as considering complex-valued channels. 

Moreover, while design criteria for the sets of vectors have been mentioned, it still remains open to describe analytically a way of finding the best possible disjoint partition of the equation coefficient vectors. 

\section*{Acknowledgements}
The authors are financially supported by the Academy of Finland grants \#276031, \#282938, \#283262 and \#268364, the Finnish Funding Agency for Innovation grant 40142/13, and a grant from Magnus Ehrnrooth Foundation. The support from the European Science Foundation under the COST Action IC1104 is also gratefully acknowledged.

\end{document}